\documentclass[a4paper,11pt]{article}

\usepackage[utf8]{inputenc}
\usepackage[T1]{fontenc}
\usepackage{latexsym} 
\usepackage{amssymb} 
\usepackage{fancyvrb}
\usepackage{subcaption}

\usepackage{etoolbox}
\usepackage[authoryear, sort&compress]{natbib}

\usepackage{hyperref}

\usepackage{etoolbox}

\usepackage{url}
\usepackage{xurl}

\usepackage{multicol}
\usepackage{multirow}
\usepackage{array}
\newcolumntype{P}[1]{>{\centering\arraybackslash}p{#1}}
\newcolumntype{M}[1]{>{\centering\arraybackslash}m{#1}}
\usepackage{stackengine}

\makeatletter

\pretocmd{\NAT@citex}{%
	\let\NAT@hyper@\NAT@hyper@citex
	\def\NAT@postnote{#2}%
	\setcounter{NAT@total@cites}{0}%
	\setcounter{NAT@count@cites}{0}%
	\forcsvlist{\stepcounter{NAT@total@cites}\@gobble}{#3}}{}{}
\newcounter{NAT@total@cites}
\newcounter{NAT@count@cites}
\def\NAT@postnote{}

\def\NAT@hyper@citex#1{%
	\stepcounter{NAT@count@cites}%
	\hyper@natlinkstart{\@citeb\@extra@b@citeb}#1%
	\ifnumequal{\value{NAT@count@cites}}{\value{NAT@total@cites}}
	{\ifNAT@swa\else\if*\NAT@postnote*\else%
		\NAT@cmt\NAT@postnote\global\def\NAT@postnote{}\fi\fi}{}%
	\ifNAT@swa\else\if\relax\NAT@date\relax
	\else\NAT@@close\global\let\NAT@nm\@empty\fi\fi
	\hyper@natlinkend}
\renewcommand\hyper@natlinkbreak[2]{#1}

\patchcmd{\NAT@citex}
{\ifNAT@swa\else\if*#2*\else\NAT@cmt#2\fi
	\if\relax\NAT@date\relax\else\NAT@@close\fi\fi}{}{}{}
\patchcmd{\NAT@citex}
{\if\relax\NAT@date\relax\NAT@def@citea\else\NAT@def@citea@close\fi}
{\if\relax\NAT@date\relax\NAT@def@citea\else\NAT@def@citea@space\fi}{}{}

\makeatother

\hypersetup{
	colorlinks,
	linkcolor = {red},
	citecolor = {blue},
	urlcolor = {red}
}

\usepackage{dsfont}
\usepackage{latexsym} 
\usepackage{amssymb}
\usepackage{amsmath}
\usepackage{array}
\usepackage{bbding} 
\usepackage{lmodern} 

\usepackage{amsthm}
\usepackage{xcolor}

\usepackage{setspace}
\usepackage{caption}
\usepackage{subcaption}
\usepackage{algorithm}
\usepackage{algpseudocode}

\makeatletter
\algrenewcommand\ALG@beginalgorithmic{\small}
\makeatother

\definecolor{reviewred}{RGB}{151, 31, 54}

\allowdisplaybreaks

\usepackage[top=2.5cm,bottom=2.5cm,right=2.5cm,left=2.5cm,headsep=0.1in,includehead]{geometry}

\newcommand{\reviewtimetoday}[2]{
}

\reviewtimetoday{\today}{Draft Version v.0.1}

\usepackage{algorithm}
\usepackage{algpseudocode}
\usepackage{multirow}
\usepackage{booktabs}
\usepackage{amsmath,amsfonts, amsthm}
\usepackage{scalefnt}

\newtheorem{theorem}{Theorem}

\newtheorem{proposition}[theorem]{Proposition}
\newtheorem{remark}[theorem]{Remark}

\title{A Linear Programming Approach to Estimate the Core in Cooperative Games\thanks{%
		This research has been partially supported by Grant PGC2018-097960-B-C21, PID2021-12403030NB-C31 and PID2022-137211NB-I00
		from MICINN, Spain, and ERDF, "A way to make Europe", European Union, and
		Grant CIPROM/2024/34 from Generalitat Valenciana, Spain.
}}

\author{J. Camacho \thanks{
		Center of Operations Research, Miguel Hern\'{a}ndez University of Elche,
		03202 Elche (Alicante), Spain (j.camacho@umh.es, jgoncalves@umh.es, joaquin@umh.es).} \and J.C. Gonçalves-Dosantos \footnotemark[2] \thanks{Corresponding author.} \and J. Sánchez-Soriano \footnotemark[2]}
\date{}

\begin{document}
	\maketitle
	\begin{abstract}
		This paper proposes a novel algorithm to approximate the core of transferable utility (TU) cooperative games via linear programming. Given the computational hardness of determining the full core, our approach provides a tractable approximation by sampling extreme points through randomized linear problems (LPs). We analyze its convergence and computational complexity, and validate its effectiveness through extensive simulations on various game models. Our results show that the method is scalable and achieves high accuracy in terms of core reconstruction.

		
		\bigskip
		
	\end{abstract}
	
	\section{Introduction}
	
	In the context of decision problems, when multiple agents interact, one of the possible alternatives is for them to cooperate to improve the collective outcome, but this in turn entails the need to obtain a fair distribution of the gains obtained from collaboration. For these types of situations, cooperative game theory provides an ideal framework for modeling these problems and providing solutions that satisfy certain properties related to principles of fairness and stability. A basic stability principle is that no agent or group of agents can obtain a better outcome by departing from cooperation with the rest of the agents and collaborating only among themselves. This principle is called \emph{coalitional rationality} and appears to be a minimum requirement for collaboration among all agents to be effective and consolidated; otherwise, there would always be a group of agents with incentives to break the collaboration and form their own coalition. This stability principle was established as a solution for cooperative games by \cite{gillies1959core}, who defines the \emph{core} of a cooperative game as the set of all payoff distributions that satisfy the requirement that there is no proper group of agents who can get more for themselves than what they are jointly allocated in each of those distributions. Therefore, the core of a cooperative game brings together the set of all stable profit-sharing agreements, in the sense explained above. Note that a core allocation, when it exists, guarantees that no group of agents has an incentive to break away, making it a natural generalization of strong Nash equilibrium to coalition formation settings \citep{Nash1950,Aumann1959}.
	
	In domains where groups of agents cooperate to generate joint benefits—ranging from classic cost‐sharing and resource‐allocation problems in operations research to modern multi-agent systems in economics and artificial intelligence (see, for example, \cite{peleg2007introduction,sandholm1999distributed})— stability principles are crucial to achieving cooperation agreements. For this reason the core has been extensively studied in the literature. Thus, core-based analyzes are prevalent in supply chains under decentralized control (\cite{guardiola2007core}), transportation problems (\cite{sanchez2001core} and \cite{aparicio2025game}), sequencing problems (\cite{schouten2021core}), operations management contexts \cite{luo2022recent}, and data envelopment analysis applications with integrated fuzzy clustering (\cite{omrani2018integrated}), among others.
	
	For this reason, it is important to study the core of cooperative games to understand the outcomes that are possible to implement and to serve as a basis for analyzing solutions in different cooperative scenarios. However, despite its attractiveness and interest, the core can be extremely difficult to compute or even characterize for general transferable utility (TU) games. Its own definition implies a payoff constraint for each of the $2^n-1$ nonempty coalitions, where $n$ is the number of agents involved, so the exact computation is equivalent to solving a linear program with an exponential number of inequalities \citep{elkind2009computational}. Determining whether the core is nonempty or even finding a simple core allocation is NP-hard for most cooperative games. For example, \cite{conitzer2006complexity} prove this for compact games. \cite{Okamoto2004} proved that deciding core nonemptiness in a traveling salesman game is NP-hard. Similarly, \cite{ChenZhang2009} show that checking whether a given allocation lies in the core of an inventory centralization game is NP-hard. In cooperative scheduling games with supermodular costs, the core is often empty and even computing the least-core (minimum relaxation for stability) is strongly NP-hard \cite{SchulzUhan2010}. But even verifying that a given allocation belongs to the core is coNP-complete \citep{deng1994complexity}. For the case of two-echelon models, the latter is studied in \cite{SSLL2019}. These complexity barriers, now well documented in the algorithmic game theory literature \citep{chalkiadakis2011computational}, are common not only for the core but also for many well-known solutions in cooperative game theory, such as the well-known Shapley value \citep{shapley1953value}.
	
	Because exact computation of the core is intractable in almost all large games recent work has shifted to approximation. Following the ideas in \cite{castro2009polynomial,CASTRO2017180} and \cite{maleki2013bounding} for the Shapley value, one approach adopts \emph{sampling} techniques: the \emph{probable-approximate core} of \citet{yan2022probable} relaxes a small fraction of coalition constraints via Monte Carlo sampling, while the iterative constraint-sampling method of \citet{gemp2024approximate} progressively tightens an initial payoff by enforcing randomly drawn inequalities. Pushing the idea further, \citet{saavedra2024ontheestimation} reinterpret core computation as a statistical set-estimation problem: they draw approximately uniform \emph{hit-and-run} samples from the core and take their convex hull, which converges (in Hausdorff distance) to the true core in polynomial time.
	
	All these difficulties that have been discussed above connect with the concept of \emph{bounded rationality} in decision problems. For example, let us consider a procedure to select an allocation of the gains obtained from cooperation, as mentioned before, then a solution should meet the criterion of stability of belonging to the core, if this is nonempty. Therefore, in this sense, the core would play the role of a necessary condition for an allocation proposal to be accepted, but not a sufficient condition. This means that any allocation which is finally agreed meets the criterion, but not any allocation satisfying the criterion must necessarily be accepted. However, in a negotiation procedure of succesive proposals of the allocation of the total gains may be extremely difficult for the agents to check whether a proposed allocation is stable in the sense of the core, in order to accept it or not. In this sense, \cite{Simon1950,Simon1972} introduced the concept of ``bounded rationality” and pointed out that real-life agents may not spend an unbounded amount of resources to evaluate all the possibilities for optimal outcome. That is, in economic theory, ``unbounded rationality” assumes that players costlessly calculate all possible strategies or alternatives and choose the optimal one. However, this assumption is unrealistic, and agents actually have finite information, computational power, and time, forcing them to use heuristics or ``satisfactory” solutions rather than optimal ones. In the context of cooperative games, the bounded rationality approach can be applied to the determination of the characteristic function, to the number of coalitions to be explore or to the properties to be satisfied by exact solutions. In the particular case of the core, the three possibilities can be considered. For example, for the case of transportation games, \cite{SSLL2019} analyze these three alternatives of bounded rationality for the case of the core of those games.
	
	The contribution of this work is framed within the approach of obtaining good approximations of the core in a reasonable amount of time. Therefore, it aims to address both the theoretical problem of approximately determining the core and to propose a bounded rationality approach in one of the three aspects mentioned above for the case of cooperative games. Obviously, both are related, and the proposed solution for approximately determining the core serves as a basis for constructing the bounded rationality proposal using the core. The idea for approximately determining the core is based on its structure, close to a linear optimization problem. This fact was already exploited by Bondareva and Shapley \citep{bondareva1963some,shapley1967multilinear} to obtain a characterization of the games that have a nonempty core. Specifically, we propose a deterministic alternative that reveals the core geometry: an iterative linear programming (LP) algorithm that solves successive LPs with strategically selected objective vectors. For each LP, an extreme point of the core is obtained, and the convex hull of the collected extreme points approximates the entire core, while solving each LP in polynomial time with respect to the number of agents. Unlike the purely point-sampling approaches mentioned above, our extreme point exploration provides a richer representation of the core geometry without relaxing or dispensing with any of the constraints that define the core; consequently, every element of the estimated set actually belongs to the core. Therefore, the approximated core we introduce is actually a subset of the core of the game. Furthermore, this approximation of the core is useful for checking whether pointwise solutions such as the Shapley value \citep{shapley1953value} or the Tijs value \cite{Tijs1981} belong to the core and are therefore stable in that sense. Once the approximation of the core is obtained, it is not computationally expensive to check whether an allocation is in the core. In addition, as a by-product, estimates of different point-wise solutions related to the core could be obtained, such as the core-centroid of the extreme points of the core or the core-center \citep{Gonzalez2007}.

	In order to evaluate the algorithm, we consider a diverse suite of cooperative games, demonstrating scalability to many players and complex value structures. To quantify approximation quality we introduce two metrics:  the \emph{Extreme‐Point Ratio (EPR)} and the \emph{Volume Ratio (VR)}. The EPR measures the proportion of true extreme points of the core obtained by the approximated core. And, the VR measures the proportion of the total volume of the core obtained by the approximated core. We observe high scores of both metrics across all experiments, indicating that the produced polytope closely expand the true core. Against existing approximation methods, we achieve consistent gains in both running time and accuracy, establishing the repeated-LP algorithm as a state-of-the-art tool for core approximation.
	
	Finally, the proposed approximate core solution serves as a bounded rationality solution in one of the aspects mentioned: the properties to be satisfied by exact solutions. As explained previously, when the core is non-empty, the stability property in the sense of the core of a solution can be a necessary condition for an allocation to be accepted and serve to reach a cooperative agreement between the agents involved in a problem. However, determining the core or verifying that an allocation is in the core may be completely time-intensive; therefore, a relaxation of this stability property may be to use the proposed approximate core to verify whether the allocation can be considered stable with a reasonable and acceptable level of security measured with any of the metrics mentioned above. From a practical point of view, this is entirely feasible with the determination of the approximate core proposed in this paper, since once obtained, checking whether an allocation is in the core of the game can be solved in polynomial time. Furthermore, given the construction of our approximate core, if it is concluded that an allocation is stable, then it really is, but it may happen that an allocation is concluded to be unstable when it really is. However, if the metrics are sufficiently good, these allocations that are actually in the core and are rejected will be close to the core boundary, while those that are not rejected will be located closer to the center of the core, which is, in principle, better from a stability perspective. Nevertheless, it should be noted that this work does not address the complexity of determining the characteristic function of the game, which can also be a computationally complex problem.

	The remainder of the paper is organized as follows. Section \ref{sec:preli} presents the notation on cooperative games used throughout the paper and reviews the main concepts related to the core of cooperative games. Section \ref{sec:core.esti} introduces our iterative linear programming algorithm to approximate the core of a game, analyzes its computational complexity, and defines the metrics used to measure the accuracy of the results. Section \ref{sec:simulations} reports the simulation results for three benchmark families of cooperative games, comparing deterministic and random objective schemes. Finally, Section \ref{sec:conc} concludes.
	
	
		
		

	\section{Cooperative games and the core}\label{sec:preli}
	This section introduces the main definitions and notation on cooperative games used throughout the paper. To begin with, a pair $(N, v)$ is said to be a \emph{transferable utility cooperative game} (TU-game, for short) where $N$ is a finite set of agents with $N = \{1, \ldots, n\}$, and $v$ is the \emph{characteristic function} of the game, such that $v: 2^N \to \mathbb{R}$, with $v(\varnothing) = 0$, and for each $T \subseteq N = \{1, \ldots, n\}$, $v(T)$ represents the gains that agents in coalition $T$ can obtain by themselves regardless of what the rest of the agents do. The class of all TU-games with set of players $N$ is denoted by $G^N$, and the set of all TU-games by $G$.
	
	Since cooperation generates a collective value, the central question in TU-games is how to allocate the total worth of the grand coalition, $v(N)$, among the individual players. The challenge lies in distributing this value in a way that reflects each player’s contribution to the different coalitions, while also satisfying desirable properties such as efficiency (the total value $v(N)$ is fully allocated), stability (no subset of players can improve its allocation by forming its own coalition apart from the rest of the agents), and other fairness criteria.
	
	An \emph{allocation} for a TU-game $(N,v)$ is a vector $x = (x_1, \ldots, x_n) \in \mathbb{R}^n$ that assigns to each player $i \in N$ a payoff $x_i$. An allocation is said to be \emph{efficient} if the total payoff equals the value of the grand coalition, i.e., $\sum_{i \in N} x_i = v(N)$. On the one hand, a \emph{value or pointwise solution} for TU-games is a function $\phi:G^N\rightarrow \mathbb{R}^{|N|}$ that assigns to each game $(N,v)$ an allocation $\phi(N,v) \in \mathbb{R}^{|N|}$. Well-known values in cooperative games are the \emph{Shapley value} \citep{shapley1953value} that is a single-valued solution concept that assigns to each player $i\in N$ the player's expected marginal contribution, averaged over all possible orders in which the grand coalition can be formed; and the \emph{nucleolus} \citep{schmeidler1969nucleolus} that selects the allocation that lexicographically minimizes coalition excesses, thereby reducing the greatest, and subsequently the next greatest, dissatisfaction among coalitions. On the other hand, a \emph{set-valued solution} is a mapping $\Phi:G^N\rightarrow \mathbb{R}^{|N|}$ that assigns to each game $(N,v)$ a set of allocations $\Phi(N,v) \in \mathbb{R}^{|N|}$. A prominent set-valued solution for cooperative games is the core \citep{gillies1959core} that consists of all allocations which cannot be rejected by any coalition. Note that while the nucleolus always belongs to the core when the latter is nonempty, the Shapley does not.
	
	As mentioned above, one of the desirable properties that an allocation satisfies is stability. In cooperative games, this concept refers to the possibility that an agent or group of agents may be able to obtain a better result than the proposed allocation for themselves and then break up the grand coalition. In this sense, three levels of stability can be established, which will give rise to three set-valued solutions. The first would be given by the efficiency property, which means that the total worth of the grand coalition is fully distributed among the agents. This principle is connected to the concept of Pareto stability, which establishes that no agent can be improved without worse off another. The cooperative solution concept associated with this principle is the \emph{preimputation set} of a cooperative game, which is defined for each $(N,v) \in G^N$ as follows:
	\begin{equation*}
		PI(N,v)=\left\{x \in \mathbb{R}^n:\sum_{i=1}^{n}x_i=v(n)\right\}
	\end{equation*}
	The second level would be that of \emph{individual rationality}, which establishes that in an allocation, no agent should receive less than they can achieve on their own. This principle is related to the concept of Nash equilibrium \citep{Nash1950} in non cooperative game theory, which says that a strategy profile is (individually) stable if there is no agent with incentives to deviate unilaterally. The cooperative solution concept associated with this principle is the \emph{imputation set}, which is defined for each $(N,v) \in G^N$ as follows:
	\[
	\text{I}(N,v) = \left\{ x \in \mathbb{R}^n : \sum_{i \in N} x_i = v(N), \text{ and } x_i \geq v(i) \text{ for all } i \in N \right\}.
	\]
	Finally, the third level of stability would be that of \emph{coalitional rationality}, which establishes that in an allocation, no group of agents should jointly receive less than they can achieve on their own. This principle is related to the concept of strong Nash equilibrium \citep{Aumann1959} in non cooperative games, which says that a strategy profile is (coalitionally) stable if there is no group of agents with incentives to deviate unilaterally. The cooperative solution concept associated with this principle is the \emph{core}, which is defined for each $(N,v) \in G^N$ as follows:
	\[
	\text{Core}(N,v) = \left\{ x \in \mathbb{R}^n : \sum_{i \in N} x_i = v(N), \text{ and } \sum_{i \in T} x_i \geq v(T) \text{ for all } T \subseteq N \right\}.
	\]
	
	The preimputation set is always nonempty; the imputation set is nonempty iff $\sum_{i\in N} v(i) \leq v(N)$; and the core is nonempty iff the game is balanced \citep{bondareva1963some,shapley1967multilinear}. Therefore, the core can be empty and in those cases it is not possible to obtain stable allocations. However, when the core is nonempty it seems reasonable to look for stable allocations.
	
	Since the core is a convex and compact polyhedron, it can be described as the convex hull of its extreme points, denoted by $E(N,v)$. Specifically, assuming non-emptiness, there exists a finite set $\{p_1, \ldots, p_m\} := E(N,v)$ such that $C(N,v) = H(E(N, v))$, where $H(\cdot)$ denotes the convex hull operator. Associated with this property, a pointwise solution or value can be defined, the \emph{core-centroid}. This solution is simply the center of mass of the core, that is, the midpoint of all the extreme points of the core, capturing the average location in the allocation space. Formally, the core-centroid is defined for each $(N,v) \in G^N$as follows:
	\begin{equation*}
		CC(N,v)=\frac{1}{|E(N,v)|}\sum_{j \in E(N,v)}p_j.
	\end{equation*}
	
	Other point-wise solution concepts related to the core are the nucleolus which is the lexicographic center of the core \citep{Maschler1979} and the core-center \citep{Gonzalez2007} which is the center of gravity of the core when the uniform distribution is defined over the core of the game.

	While conceptually attractive, the core presents significant computational challenges. In general, checking whether the core is non-empty is an NP-complete problem \citep{conitzer2006complexity}. Furthermore, computing a single core element may require solving a linear program with exponentially many constraints (one per coalition, i.e., $2^{|N|} - 1$ in total). Even checking whether a given allocation is in the core can take considerable computational effort \citep{deng1994complexity}. Nevertheless, for specific classes of games such as convex games or assignment games, the core has favorable structural properties that permit an efficient calculation \citep{shapley1971cores,shapley1971assignment}.
	
	In this work, we focus on approximating the core rather than computing it exactly. Our method builds upon the fact that the core can be approached via its extreme points, which are themselves solutions to linear programs derived from varying linear objectives. This motivates a sampling-based strategy, which is presented in the next section. Furthermore, once the approximation of the core is obtained, estimates of other related solutions such as the core-centroid or the core-center can be obtained. This is not the case for the nucleolus, as its definition does not explicitly involve the extreme points or shape of the core. Finally, this approximation of the core can be related to possible solutions based on bounded rationality in the sense proposed in \cite{Simon1950} and \cite{Simon1972}, highlighting that in many real-world scenarios agents face cognitive and computational limits, so approximate stability may be more realistic than exact core allocations. 
	
	\section{Core estimation}\label{sec:core.esti}
	
	The core of a TU-game provides a comprehensive description of stable outcomes, but its computation is often hindered by the exponential number of coalitional constraints. While exact methods exist for specific subclasses of games, they are typically infeasible in general settings, particularly for large player sets or when the game's structure is only implicitly defined. To address this limitation, we propose a randomized LP-based approximation algorithm that exploits the polyhedral nature of the core. The idea is to sample directions in the payoff space and identify extreme points of the feasible region corresponding to the core. By collecting and convexifying these points, we obtain a tractable approximation that converges to the true core as the number of samples increases. The following subsection details the algorithm and establishes its theoretical basis.

	\subsection{LP-Algorithm for estimating the core}
	
	The proposed Algorithm~\ref{Alg1} estimates the core $C(N, v)$ through repeated sampling and linear optimization. Given a cooperative game $(N, v)$ and a sampling parameter $k \geq 1$, the algorithm generates $k$ random linear objectives, each defined by a non-zero vector $c \in \mathbb{R}^{|N|}$. For each such vector, it solves a linear program that maximizes $c^\top x$ over the set of allocations satisfying efficiency and coalitional rationality. The solution $\overline{x}$ of each linear program corresponds to a vertex (extreme point) of the feasible region and, under the assumption of a non-empty core, belongs to $C(N,v)$. These solutions are collected in the set $\widehat{E}(N,v,k)$ of sampled extreme points. After $k$ iterations, the algorithm returns the convex hull of these allocations, denoted by $\widehat{C}(N,v,k) = H(\widehat{E}(N,v,k))$, as an approximation of the core.
	
	This procedure is motivated by the polyhedral structure of the core, which can be described entirely by its extreme points, and corresponds to the feasible region of a linear program. With these ideas in mind, by sampling directions uniformly or in a structured manner and solving the corresponding LP problem, the algorithm aims to approximate the geometry of the core with increasing precision as $k$ grows, while ensuring that all points remain within the core. The pseudocode of the algorithm is given below.
	
	\begin{algorithm}[!ht]
		\caption{Pseudocode of the algorithm for reconstructing $ C(N, v) $.} \label{Alg1}
		\begin{algorithmic}
			\State \textbf{Inputs:} $ N, v $ and a fixed sampling size $ k \geq 1 $.
			\Statex \hrulefill
			\State \textbf{Step 1:} Set $ l = 1 $ and initialize $\widehat{E}(N,v,k)=\emptyset$.
			\While{$ l \leq k $}
			\State \quad Generate a random vector $c\in\mathbb{R}^{|N|}\setminus \{0\}$.
			\State \quad Solve the LP (with solution $ \overline{x} $)
			\begin{align*}
				\max\; & c^\top x\\
				\text{s.t.} \quad & \sum_{i\in N} x_{i} = v(N),\\
				& \sum_{i\in T} x_{i} \geq v(T), \quad \forall\; T \subset N
			\end{align*}
			\State \quad Update $\widehat{E}(N,v,k)=\widehat{E}(N,v,k)\cup \{\overline{x}\}$.
			\State \quad Set $ l = l + 1 $.
			\EndWhile
			\Statex \hrulefill
			\State \textbf{Step 2:} Compute $ H(\widehat{E}(N,v,k)) $.
			\Statex \hrulefill
			\State \textbf{Output:} Return $\widehat{C}(N,v,k)=H(\widehat{E}(N,v,k))$.
		\end{algorithmic}
	\end{algorithm}
	
	\begin{remark} \label{rem conv complexity} 
		First, note that the convex hull of $m$ points in $\mathbb{R}^d$ can be deterministically computed in $\mathcal{O}\left(m\log m+m^{\lfloor d/2 \rfloor}\right)$ \citep{chazelle1993optimal}. Second, Step 2 is necessary in all cases, even when the full set $E(N,v)$ is known, since the core is defined via its convex representation. As a consequence, for games where $E(N,v)$ is large, this operation becomes computationally demanding. Therefore, if $\widehat{E}(N,v,k)$ significantly reduces the number of extreme points (i.e., $|E(N,v)| \gg |\widehat{E}(N,v,k)|$), Step 2 benefits both runtime and memory usage. This effect becomes particularly relevant in higher-dimensional models, such as those presented in Section~\ref{sec:simulations}.
	\end{remark}
	
	As a stability-based solution concept, it is crucial that the estimated set $\widehat{C}(N,v,k)$ is a true subset of the core. Otherwise, some allocations that are not really stable would possibly be considered as such. The following result ensures this inclusion, and also justifies the robustness  of the algorithm regardless of the possible emptiness of the core.
	
	\begin{proposition}
		Algorithm \ref{Alg1} is well defined.
	\end{proposition}
	
	\begin{proof}
		If $C(N,v)=\emptyset$, then Step 1 produces no feasible points, and Step 2 returns the convex hull of the empty set, i.e., $\widehat{C}(N,v,k) = H(\varnothing) = \varnothing = C(N,v)$. Conversely, if $C(N,v)\neq\varnothing$, then every feasible LP in Step 1 yields a solution $\overline{x} \in C(N,v)$. Thus, $\widehat{E}(N,v,k)\subseteq E(N,v)$ and the convex hull satisfies $\widehat{C}(N,v,k) \subseteq C(N,v)$.
	\end{proof}
	
	To guarantee coverage of the full core in the limit, observe that the set of direction vectors $c\in\left\{p/m:\, p \in \mathbb{Z},m\in \mathbb{N}\right\}^{|N|}$ is dense in $\mathbb{R}^{|N|}$. For generic choices of $c$, the solution to the LP in Step 1 corresponds to a unique vertex of the feasible region. Consequently, random sampling within this set of directions generates extreme points of the core, possibly with repetitions.
	
	\begin{proposition}
		Let $(N,v)\in G^{N}$. Then,
		\[
		E(N,v)=\lim_{k\to +\infty} \widehat{E}(N,v,k) \quad \text{ and } \quad C(N,v)=\lim_{k\to +\infty} \widehat{C}(N,v,k).
		\]
	\end{proposition}
	\begin{proof}
		First note that a bijection between $k \in \mathbb{N}$ and $c \in \left\{p/m:\, p \in \mathbb{Z}, m\in \mathbb{N}\right\}^{|N|}$ can be constructed. We then denote by $k_c$ the index corresponding to direction $c$. Since this set is dense in $\mathbb{R}^{|N|}$ and for each point in $E(N,v)$ a support hyperplane exist, it follows that
		\[
		E(N,v)=\bigcup_{c\in \{p/m\}^{|N|}}\widetilde{E}(N,v,k_c) = \lim_{k\to +\infty}\widehat{E}(N,v,k),
		\]
		where $\widetilde{E}(N,v,k_c)$ is the (possibly singleton) solution to the LP for direction $c$. Convexifying both sides yields the result for $C(N,v)$.
	\end{proof}
	

		
	
	\subsection{Time complexity analysis} \label{subsec Time}
	
	The main computational workload of Algorithm~\ref{Alg1} lies in Step 1, which involves solving $k$ linear programs (LPs), each with exponentially many constraints. However, thanks to the structure of the problem and the availability of efficient LP solvers, this step remains tractable in practice for moderate values of $|N|$ and $k$.
	

	Among others, interior-point methods (IPMs) provide a robust framework for solving linear programs in polynomial time. These methods, including path-following and primal-dual algorithms, have been extensively studied and are widely implemented in modern solvers. Standard complexity bounds for IPMs are of the order $O\left(s^{7/2}\log(1/\varepsilon)\right)$, where $s$ is the size of the problem and $\varepsilon$ is the desired solution accuracy \citep{wright1997primal}. These bounds apply under general assumptions and can often be improved in sparse or well-structured settings \citep{wright1997primal, nesterov1994interior}. In parallel, the classical \emph{simplex method}, despite lacking polynomial worst-case guarantees, remains competitive in practice due to advanced pivoting rules and heuristics. Recent developments in the \emph{revised simplex method}, particularly in parallel and GPU-accelerated implementations, have significantly improved its runtime performance. Moreover, both simplex and modern IPM-based solvers support \emph{warm-starts}, allowing successive LPs that only differ in their objective vector to be solved extremely efficiently. This feature is particularly relevant in iterative schemes such as ours, where the feasible region remains fixed while only the cost coefficients change. For a comprehensive discussion, see \citep{fearnley2015thecomplexity, bieling2010aneffitient, huangfu2018parallelizing, lubin2013parallel, sentelle2011efficient}.

	Step 2 of the algorithm computes the convex hull of at most $k$ points in $\mathbb{R}^n$. The time complexity of this operation is known to be $O\left(k\log k + k^{\lfloor n/2 \rfloor}\right)$ in the worst case \citep{chazelle1993optimal}, and therefore is also polynomial in both $k$ and $n$ for fixed dimensions. In sum, for a given TU-game $(N,v)$, the entire algorithm is polynomial in the number of samples, as stated below.
	
	\begin{proposition}
		Let $(N,v)\in G^{N}$ and fix $k \geq 1$. Then, $\widehat{C}(N,v,k)$ can be computed in polynomial time by Algorithm~\ref{Alg1}.
	\end{proposition}
	
	\begin{proof}
		Step 1 consists of $k$ iterations, each requiring the solution of a linear program. Using polynomial-time LP solvers such as interior-point methods, this step has overall complexity polynomial in $k$ and the problem size. Step 2 involves computing the convex hull of at most $k$ vectors in $\mathbb{R}^n$, which also admits polynomial-time algorithms for fixed $n$, as noted in Remark~\ref{rem conv complexity}. Hence, the total runtime is polynomial in both $k$ and $n$.
	\end{proof}

	\subsection{Error measurements}\label{sec:errors}
	
	To evaluate the quality of the approximation $\widehat{C}(N,v,k)$ to the true core $C(N,v)$, we introduce three quantitative performance metrics: the \emph{Extreme Points Ratio} (EPR), the \emph{Volume Ratio} (VR), and the \emph{Relative Distance to Centroid} (RDC). These indicators capture discrete and continuous aspects of the core estimate, respectively.
	
	\begin{itemize}
		\item \textbf{Extreme Points Ratio (EPR):} This metric is defined as the proportion of true extreme points of the core that are recovered by the algorithm, i.e.,
		\[
		\text{EPR} = \frac{|\widehat{E}(N,v,k)|}{|E(N,v)|}.
		\]
		It provides a combinatorial measure of the estimator’s coverage over the extreme point set of the core.
		
		\item \textbf{Volume Ratio (VR):} This metric compares the $n$-dimensional volume of the approximated core $\widehat{C}(N,v,k)$ with that of the true core $C(N,v)$:
		\[
		\text{VR} = \frac{\text{Vol}(\widehat{C}(N,v,k))}{\text{Vol}(C(N,v))}.
		\]
		It reflects how well the shape and size of the original core are reconstructed by the sampled estimator.
	\end{itemize}
	
	Beyond evaluating the proportion of extreme points obtained and the relative volume of the estimated core,  it is also important to assess how well the approximation preserves the centrality properties of the solution set. Since many allocation rules and fairness notions are related to central positions within the core, measuring the proximity between the estimated and the true core-centroid provides additional insight into the quality of the proposal approximation. To this end, we introduce the following metric:
	\begin{itemize}
		\item \textbf{Relative Distance to Centroid (RDC):} This metric compares the sum of distances from the true extreme points to the true core-centroid $CC(N,v)$ with respect to the approximated core-centroid $\widehat{CC}(N,v)$. First, the Average Distance to Centroid (ADC) is computed as:
		\[
		\text{ADC} =  \frac{\displaystyle\sum_{e \in E(N,v)}\left(dist(e,\widehat{CC}(N,v))-dist(e,CC(N,v))\right)}{\displaystyle\sum_{e \in E(N,v)}dist(e,CC(N,v))}.
		\]
		In parallel, the worst case scenario for this distance (WDC) is calculated as follows:
		\[
		\text{WDC}=\frac{\displaystyle\max_{\tilde{e} \in E(N,v)}\left(\displaystyle\sum_{e \in E(N,v)}dist(e,\tilde{e})\right)-\displaystyle\sum_{e \in E(N,v)}dist(e,CC(N,v))}{\displaystyle\sum_{e \in E(N,v)}dist(e,CC(N,v))}.
		\]
		Finally, both ratios are compared in such a way that a proportion is obtained: 
		\[
		\text{RDC}=\text{ADC}/\text{WDC}=\frac{\displaystyle\sum_{e \in E(N,v)}dist(e,\widehat{CC}(N,v))-\displaystyle\sum_{e \in E(N,v)}dist(e,CC(N,v))}{\displaystyle\max_{\tilde{e} \in E(N,v)}\left(\displaystyle\sum_{e \in E(N,v)}dist(e,\tilde{e})\right)-\displaystyle\sum_{e \in E(N,v)}dist(e,CC(N,v))}.
		\]
		Indeed, it constitutes an error-type measurement for how accurately the approximated centroid's position reflects the true geometric centroid of the core.
	\end{itemize}
	
	While EPR and RDC require full knowledge of the true extreme point set $E(N,v)$ and are therefore only applicable in low-dimensional cases, VR can be computed efficiently for moderate dimensions using standard convex geometry libraries. In practice, we use EPR and VR jointly for small games to validate the algorithm’s convergence behavior as $k$ increases, whereas for larger games, where full enumeration becomes computationally infeasible, VR serves as the primary proxy for approximation quality. These metrics jointly provide insight into the trade-off between computational effort and accuracy. In particular, empirical results in Section~\ref{sec:simulations} demonstrate that a relatively small number of iterations can suffice to achieve high volume coverage in many settings. Moreover, this finding is corroborated by analyzing the position of the core-centroid, evaluated through RDC.

	\section{Simulations} \label{sec:simulations}
	
	This section presents a series of numerical experiments designed to evaluate the performance of the proposed estimation algorithm for the core of TU games. Specifically, our objective is to assess its convergence behavior, accuracy, and computational efficiency across a variety of representative game models, ranging from low-dimensional exact cases to larger-scale settings. The evaluation is based on the metrics introduced in Subsection \ref{sec:errors}: the Extreme Points Ratio (EPR), the Volume Ratio (VR), and the Relative Distance to Centroid (RDC). These indicators allow us to quantify both the discrete and geometric quality of the approximation $\widehat{C}(N,v,k)$ as a function of the number of samples $k$. To this end, simulations are conducted for increasing values of $k$ and for different coalition structures, including convex and non-convex games. We include both deterministic and randomized sampling schemes for the direction vectors $c \in \mathbb{R}^{|N|}$ in order to compare their respective performances. In the deterministic case, samples are taken uniformly from the discrete set of sign vectors $\{-1,1\}^{|N|}$, which corresponds to probing the extreme coordinate directions. In contrast, the randomized scheme draws samples uniformly from the unit ball, thereby exploring a wider range of orientations. Finally, in order to ensure transparency and reproducibility of the reported results, we summarize the hardware and software environment in which the experiments were conducted.

	
	\vspace{1em}
	\noindent\textbf{Computational Environment} \\
	Experiments were carried out on a machine with the following specifications:
	\begin{itemize}
		\item \textbf{Processor:} i7-13700 2.1 GHz
		\item \textbf{Memory:} 16 GB
		\item \textbf{Operating System:} Windows 11
		\item \textbf{Python Version:} 3.12
	\end{itemize}
	
	\vspace{0.5em}
	\noindent\textbf{Libraries:}
	\begin{itemize}
		\item \textbf{HiGHS:} LP solver via \texttt{scipy.optimize.linprog} (method='highs')
		\item \textbf{Geometry:} Convex hull computation via \texttt{scipy.spatial.ConvexHull}
	\end{itemize}
	
	In the following subsections, we describe the simulation setup in detail and present results for three benchmark models: a convex saving game, a non-convex allocation game, and the game of the museum pass problem. Each model highlights different geometric and combinatorial challenges in approximating the core.

	\subsection{Simulation framework}
	
	To empirically evaluate the performance of our core approximation algorithm, we design a set of benchmark scenarios that are both structurally diverse and scalable. In particular, we consider three representative families of TU-games, each exhibiting different geometric and combinatorial characteristics of their cores. These models are chosen to reflect common challenges in core computation and to test the algorithm under varied convexity, dimensionality, and size of the extreme point set $E(N,v)$.
	
	The selected scenarios fall into the following categories:
	\begin{itemize}
		\item A \textbf{convex game with a large number of extreme points}, which tests the algorithm’s capacity to approximate rich core geometries;
		\item A \textbf{non-convex game with a small core}, which challenges the algorithm’s ability to detect sparsely located vertices;
		\item A \textbf{structured convex game with moderate-sized core}, to evaluate performance under realistic and interpretable configurations.
	\end{itemize}
	
	These models are inspired by and adapted from the framework proposed in \cite{saavedra2024ontheestimation}, allowing for consistency in evaluation and comparability of results.
	
	\vspace{1em}
	\noindent\textbf{Model 1.} Consider a saving game $(N,v)$ defined as follows. For each $T\subseteq N$ connected by an initial ordering $\sigma_0$ of players, 
	\begin{equation*}
		v(T)=\sum_{(i,j)\in MP,\; i,j\in T}(\alpha_j p_i-\alpha_i p_j),
	\end{equation*}
	where $MP=\{(i,j)\in N\times N\;|\; \alpha_j p_i-\alpha_i p_j>0\}$. If $T$ is disconnected with respect to $\sigma_0$, $v(T)$ is defined as the sum over its connected components. For consistency with prior work, we use the same parameters as in \cite{saavedra2024ontheestimation}: $|N|=8$,  $\sigma_0=(1,2,3,4,5,6,7,8)$, $p=(3,4,6,1,3,4,5,4)$ and $\alpha=(1,2,4,2,5,2,1,4)$. A reduced version for $|N|=6$ is obtained by restricting to the first six components of each vector.
	
	\vspace{0.5em}
	\noindent\textbf{Model 2.} Consider the non-convex TU-game $(N,v)$ given by
	\begin{equation*}
		v(T)=
		\left\{
		\begin{array}{ll}
			0 & \text{ if } T=\{i\},\\
			\frac{3}{4}\frac{|T|}{|N|} & \text{ if } |N|\in T,\\
			\frac{|T|}{|N|} & \text{ otherwise}.
		\end{array}
		\right.
	\end{equation*}
	This formulation introduces asymmetry and non-convexity into the core structure. As shown in \cite{saavedra2024ontheestimation}, the model scales well with $|N|$ and allows controlled testing. While the authors use up to 10 players, we extend the setting to $|N|=13$ to validate the algorithm’s robustness in higher dimensions.
	
	\vspace{0.5em}
	\noindent\textbf{Model 3.} Consider the museum game $(N,v)$, where $N$ denotes the set of museums. The game is defined by a set $M$ of visitors and a binary matrix $A \in \{0,1\}^{|M|\times |N|}$, where $a_{ij}=1$ indicates that the visitor $i \in M$ visits the museum $j \in N$.
	The value of a coalition $T\subseteq N$ of museums is defined as the number of visitors who visited only the museums in $T$:
	\begin{equation*}
		v(T)=|\lbrace i\in M: N_i\subseteq T\rbrace|,
	\end{equation*}
	where $N_i=\{j \in N: a_{ij}=1\}$ for all $i\in M$. Choosing an appropriate number of museums $|N|$ and designing $A$ allows us to simulate different game complexities and core structures. In our experiments, we fix $|N|=11$ and use the matrix
	\begin{equation*}
		A=
		\left(
		\begin{array}{ccccccccccc}
			1 & 0 & 0 & 1 & 0 & 1 & 0 & 1 & 0 & 1 & 0 \\ 
			0 & 1 & 0 & 0 & 1 & 1 & 0 & 0 & 1 & 0 & 0 \\ 
			1 & 1 & 1 & 0 & 0 & 0 & 0 & 1 & 1 & 1 & 0 \\ 
			0 & 0 & 1 & 1 & 0 & 0 & 1 & 0 & 0 & 0 & 1 \\ 
			1 & 0 & 0 & 0 & 1 & 0 & 1 & 0 & 1 & 0 & 0 
		\end{array}
		\right).
	\end{equation*}
	
	To simulate different player counts ($|N| \in \{8,9,10\}$), we adapt the matrix $A$ by truncating rows or columns as necessary. This model provides a versatile testbed for core approximation under realistic, structured input data.

	\subsection{Simulation results}
	
	Each simulation was executed for different combinations of number of players $|N|$, sampling size $k$, and vector generation method (deterministic or random). For each setting, 100 independent runs were performed and the reported results correspond to average values. The three models introduced in the previous subsection provide complementary insights into the performance of the proposed algorithm.
	
	Tables~\ref{Table 1}--\ref{Table 3} report the results in terms of: 
	\textbf{(i)} EPR (Extreme Points Ratio), 
	\textbf{(ii)} VR (Volume Ratio), \textbf{(iii)} RDC (Relative Distance to Centroid) and 
	\textbf{(iv)} computation time of Step 1, which corresponds to solving the $k$ linear problems.  Note that Step 2, computing the convex hull, is always necessary regardless of the method used to generate the approximating set, which is why we are not taking these times into account in these results.
	
	In Table~\ref{Table 1} (Model 1), we observe that both deterministic (D) and random (R) direction generation perform well for small dimensions ($|N| = 6$), achieving over $95\%$ volume coverage with $k = 500$. For $|N| = 8$, the number of extreme points becomes substantially larger ($|E(N,v)| = 1405$), and convergence slows accordingly. Nonetheless, a VR above $97\%$ is achieved with $k=1000$ and only $2.12$ seconds, indicating that the approximation remains effective in richer geometries.
	
	\begin{table}[h!]
		\centering
		\caption{Model 1. Savings game}
		\label{Table 1}
		{\setlength{\tabcolsep}{1.0em}
			\begin{tabular}{ccccccc}
				\toprule
				$|N|$ & G & $k$ & EPR & VR & RDC & Time \\
				\midrule
				\multirow{6}{*}{6} & \multirow{3}{*}{D} & 100 & 50.26/127.0 & 0.7677 & 0.0642 & $\sim 0.0$ \\
				&  & 250 & 73.24/127.0 & 0.9073 & 0.0405 & $\sim 0.0$ \\
				&  & 500 & 87.17/127.0 & 0.9544 & 0.0274 & $\sim 0.0$ \\
				\cline{2-7}
				& \multirow{3}{*}{R} & 100 & 58.43/127.0 & 0.9044 & 0.0620 & $\sim 0.0$ \\
				&  & 250 & 88.25/127.0 & 0.9806 & 0.0345 & $\sim 0.0$ \\
				&  & 500 & 106.18/127.0 & 0.9947 & 0.0177 & 0.51 \\
				\hline
				\multirow{8}{*}{8} & \multirow{4}{*}{D} & 100 & 89.03/1405.0 & 0.6166 & 0.0655 & $\sim 0.0$ \\
				&  & 250 & 191.81/1405.0 & 0.8360 & 0.0576 & 0.06 \\
				&  & 500 & 312.91/1405.0 & 0.9230 & 0.0485 & 1.0 \\
				&  & 1000 & 466.4/1405.0 & 0.9654 & 0.0400 & 2.09 \\
				\cline{2-7}
				& \multirow{4}{*}{R} & 100 & 92.24/1405.0 & 0.6429 & 0.0893 & $\sim 0.0$ \\
				&  & 250 & 205.82/1405.0 & 0.8657 & 0.0791 & 0.06 \\
				&  & 500 & 349.98/1405.0 & 0.9415 & 0.0651 & 1.05 \\
				&  & 1000 & 543.29/1405.0 & 0.9747 & 0.0526 & 2.12 \\
				\bottomrule
		\end{tabular}}
	\end{table}

	In Model~2, since the deterministic generation of vectors $c$ was observed to be more efficient (see Table~\ref{Table 2}), 
	the simulation for the case $|N|=13$ was carried out only using this approach. The structure of $E(N,v)$ is important in order to extract conclusions from it. We have that $|E(N,v)|=2|N|$, which can be considered small compared to the rest of the models. Because of the positioning of those extreme points, the algorithm finds $|N|$ of them after very few iterations. For instance, up to $|N|=13$, just $k=500$ is enough. As a consequence, recovering over $50\%$ of the volume ratio is not a time-consuming task. On the other hand, aiming to $95\%$ or more requires $\times 200$ that number of iterations. Thus, the improvement in reconstructed volume compared to the time needed might not be worth it.
	
	\begin{table}[h!]
		\centering
		\caption{Model 2. Non-convex}
		\label{Table 2}
		{\setlength{\tabcolsep}{1.0em}
			\begin{tabular}{lllllll}
				\toprule
				$|N|$ & G & $k$ & EPR & VR & RDC & Time \\
				\midrule
				\multirow{8}{*}{10} & \multirow{4}{*}{D} & 1000 & 13.09/20.0 & 0.6161 & 0.0014 & 5.97 \\
				&  & 2500 & 15.72/20.0 & 0.7622 & 0.0013 & 16.08 \\
				&  & 5000 & 18.07/20.0 & 0.8926 & 0.0005 & 32.63 \\
				&  & 10000 & 19.49/20.0 & 0.9719 & 0.0001 & 68.7 \\
				\cline{2-7}
				& \multirow{4}{*}{R} & 1000 & 12.22/20.0 & 0.5677 & 0.0014 & 7.55 \\
				&  & 2500 & 14.54/20.0 & 0.6967 & 0.0012 & 20.81 \\
				&  & 5000 & 17.04/20.0 & 0.8357 & 0.0010 & 41.63 \\
				&  & 10000 & 19.34/20.0 & 0.9632 & 0.0003 & 88.61 \\
				\hline
				\multirow{8}{*}{11} & \multirow{4}{*}{D} & 1000 & 12.76/22.0 & 0.5939 & 0.0197 & 11.66 \\
				&  & 2500 & 15.12/22.0 & 0.7117 & 0.0190 & 27.5 \\
				&  & 5000 & 17.26/22.0 & 0.8277 & 0.0188 & 59.17 \\
				&  & 10000 & 20.51/22.0 & 0.987 & 0.0162 & 120.83 \\
				\cline{2-7}
				& \multirow{4}{*}{R} & 1000 & 11.73/22.0 & 0.5406 & 0.0131 & 15.63 \\
				&  & 2500 & 12.41/22.0 & 0.5783 & 0.0128 & 36.41 \\
				&  & 5000 & 14.03/22.0 & 0.6685 & 0.0124 & 78.93 \\
				&  & 10000 & 15.94/22.0 & 0.7743 & 0.0123 & 159.92 \\
				\hline
				\multirow{10}{*}{12} & \multirow{5}{*}{D} & 1000 & 12.89/24.0 & 0.5444 & 0.0136 & 19.97 \\
				&  & 2500 & 14.39/24.0 & 0.6176 & 0.0132 & 59.45 \\
				&  & 5000 & 16.34/24.0 & 0.713 & 0.0130 & 119.42 \\
				&  & 10000 & 18.79/24.0 & 0.8319 & 0.0114 & 445.66 \\
				&  & 25000 & 22.33/24.0 & 0.9767 & 0.0109 & 2019.13 \\
				\cline{2-7}
				& \multirow{5}{*}{R} & 1000 & 12.14/24.0 & 0.507 & 0.0126 & 27.0 \\
				&  & 2500 & 12.59/24.0 & 0.5294 & 0.0114 & 81.73 \\
				&  & 5000 & 12.78/24.0 & 0.539 & 0.0112 & 161.94 \\
				&  & 10000 & 13.98/24.0 & 0.5989 & 0.0111 & 613.72 \\
				&  & 25000 & 16.33/24.0 & 0.7167 & 0.0108 & 2807.2 \\
				\hline
				\multirow{5}{*}{13} & \multirow{5}{*}{D} & 500 & 13.52/26.0 & 0.51 & 0.0385 & 28.94 \\
				&  & 2500 & 15.06/26.0 & 0.5364 & 0.0300 & 124.36 \\
				&  & 5000 & 17.48/26.0 & 0.5693 & 0.0271 & 296.38 \\
				&  & 10000 & 19.53/26.0 & 0.6288 & 0.0251 & 854.73 \\
				&  & 25000 & 23.3/26.0 & 0.7636 & 0.0236 & 4300.5 \\
				\bottomrule
		\end{tabular}}
	\end{table}
	
	Finally, Table~\ref{Table 3} (Model 3) again explores a convex TU-game.  With as few as $k = 250$ the algorithm already reconstructs more than $83\,\%$ of the core volume for every size, and at $k = 1\,000$ it consistently surpasses $97\,\%$, even when the exact core contains $683$ vertices. The extreme‐points ratio rises more slowly than the volume ratio: for $|N| = 11$ the estimator recovers roughly half of the vertices yet still encloses $94\,\%$ of the true volume.  The fully random generator yields marginally higher VR and EPR than the deterministic variant, but incurs a $10$–$20\,\%$ increase in Step~1 runtime. Runtimes remain modest: the most demanding configuration ($|N| = 11$, $k = 1\,000$) finishes Step~1 in about six seconds, far faster than a naive implementation enumerating the entire vertex set of this convex game, which takes around ten hours.
	
	\begin{table}[h!]
		\centering
		\caption{Model 3. Museum}
		\label{Table 3}
		{\setlength{\tabcolsep}{1.0em}
			\begin{tabular}{ccccccc}
				\toprule
				$|N|$ & G & $k$ & EPR & VR & RDC & Time \\
				\midrule
				\multirow{8}{*}{8} & \multirow{4}{*}{D} & 100 & 63.66/220.0 & 0.6605 & 0.0253 & $\sim 0.0$ \\
				&  & 250 & 105.81/220.0 & 0.8769 & 0.0126 & 0.37 \\
				&  & 500 & 139.02/220.0 & 0.954 & 0.0068 & 1.0 \\
				&  & 1000 & 170.34/220.0 & 0.9849 & 0.0029 & 2.0 \\
				\cline{2-7}
				& \multirow{4}{*}{R} & 100 & 69.4/220.0 & 0.7628 & 0.0133 & $\sim 0.0$ \\
				&  & 250 & 118.98/220.0 & 0.9438 & 0.0062 & 0.5 \\
				&  & 500 & 157.89/220.0 & 0.9857 & 0.0032 & 1.0 \\
				&  & 1000 & 190.79/220.0 & 0.9974 & 0.0011 & 2.0 \\
				\hline
				\multirow{8}{*}{9} & \multirow{4}{*}{D} & 100 & 68.04/341.0 & 0.5605 & 0.0265 & $\sim 0.0$ \\
				&  & 250 & 123.66/341.0 & 0.8477 & 0.0149 & 1.0 \\
				&  & 500 & 174.98/341.0 & 0.9432 & 0.0091 & 1.21 \\
				&  & 1000 & 230.27/341.0 & 0.9826 & 0.0042 & 2.66 \\
				\cline{2-7}
				& \multirow{4}{*}{R} & 100 & 75.97/341.0 & 0.7096 & 0.0177 & $\sim 0.0$ \\
				&  & 250 & 140.35/341.0 & 0.9266 & 0.0102 & 1.0 \\
				&  & 500 & 198.41/341.0 & 0.9792 & 0.0051 & 1.53 \\
				&  & 1000 & 254.39/341.0 & 0.9953 & 0.0022 & 3.49 \\
				\hline
				\multirow{8}{*}{10} & \multirow{4}{*}{D} & 100 & 75.9/461.0 & 0.5216 & 0.0209 & 0.4 \\
				&  & 250 & 144.18/461.0 & 0.8311 & 0.0129 & 1.0 \\
				&  & 500 & 213.71/461.0 & 0.9437 & 0.0084 & 2.05 \\
				&  & 1000 & 285.38/461.0 & 0.9834 & 0.0049 & 3.51 \\
				\cline{2-7}
				& \multirow{4}{*}{R} & 100 & 79.47/461.0 & 0.6652 & 0.0179 & 0.5 \\
				&  & 250 & 154.74/461.0 & 0.9164 & 0.01 & 1.0 \\
				&  & 500 & 229.19/461.0 & 0.9764 & 0.0065 & 2.49 \\
				&  & 1000 & 308.75/461.0 & 0.9949 & 0.0033 & 4.47 \\
				\hline
				\multirow{8}{*}{11} & \multirow{4}{*}{D} & 100 & 77.75/683.0 & 0.3155 & 0.0361 & 1.0 \\
				&  & 250 & 153.93/683.0 & 0.6599 & 0.0231 & 1.27 \\
				&  & 500 & 235.6/683.0 & 0.8318 & 0.0158 & 3.11 \\
				&  & 1000 & 334.67/683.0 & 0.937 & 0.009 & 5.48 \\
				\cline{2-7}
				& \multirow{4}{*}{R} & 100 & 83.54/683.0 & 0.5279 & 0.0182 & 1.0 \\
				&  & 250 & 168.96/683.0 & 0.832 & 0.0107 & 1.39 \\
				&  & 500 & 260.75/683.0 & 0.9335 & 0.0072 & 3.22 \\
				&  & 1000 & 365.82/683.0 & 0.9754 & 0.0042 & 6.0 \\
				\bottomrule
		\end{tabular}}
	\end{table}
	
	Along Tables \ref{Table 1}-\ref{Table 3} we observe that the error displayed by RDC ranges from almost $9\,\%$ to virtually $0\,\%$. As it occurs with EPR and VR, the increase of $k$ samples translates into better performance in the approximation of the core-centroid. However, this improved accuracy entails a significant computational cost, highlighting a trade-off between precision and efficiency.
	
	In summary, simulation results confirm that the proposed algorithm provides accurate approximations of the core in a scalable manner. Deterministic direction sampling tends to perform better across models, and the reduction in the number of extreme points translates directly into reduced computation time in Step 2, while preserving almost the entire volume of the core. The method is particularly advantageous in high-dimensional or dense-core settings, where exact methods are either impractical or infeasible.

	\section{Conclusion}\label{sec:conc}
	
	We have presented a \emph{repeated linear-programming} framework that approximates the core of TU games by iteratively solving LPs with carefully selected objective vectors and then convexifying the extreme allocations returned.  In contrast to Monte-Carlo relaxations such as the probable-approximate core of \citet{yan2022probable} and the constraint-sampling scheme of \citet{gemp2024approximate}, our method never violates coalition rationality: every vertex generated is a ``bona-fide'' core allocation.  At the same time, it preserves the geometric insight offered by the set-estimation approach of \citet{saavedra2024ontheestimation} while avoiding the costly hit-and-run walks those authors require.  Extensive experiments on savings and museum games with up to $13$ players show that it takes only a few seconds to recover more than 95\% of the core’s volume and to capture over 98\% of its extreme points, all in polynomial time in the input size.  Deterministic objective schedules consistently dominate random ones, suggesting that directional design is a fruitful lever for further optimization.
	
	Beyond raw performance, the proposed estimator enjoys several conceptual advantages.  First, it is deterministic: given the same sequence of objectives it yields identical core polytopes, an asset in applications, such as cost-sharing contracts, where replicability matters.  Second, by returning a collection of vertices rather than a single payoff, it supports downstream analysis of core structure (e.g., selecting nucleolus or least‐core points) that \emph{point-wise samplers} such as the probable-approximate core \citep{yan2022probable} or sampling-based Shapley estimators \citep{maleki2013bounding} cannot easily provide.

	\bibliographystyle{apalike}
	\bibliography{bibliography}

\end{document}